\documentclass[10pt, a4paper, twocolumn]{IEEEtran}

\usepackage{amsmath}
\usepackage{amssymb}    
\usepackage{amsfonts}
\usepackage[english]{babel}
\usepackage{cite} 
\usepackage{multirow}
\usepackage{stfloats}    

\usepackage{algorithm}   
\usepackage{algorithmic} 

\usepackage{caption} 
\usepackage{graphicx}
\usepackage{epsfig}
\usepackage{subfigure} 
\usepackage{tablefootnote}

\usepackage{url}

\usepackage{accents}
\makeatletter
\def\widebar{\accentset{{\cc@style\underline{\mskip10mu}}}}
\def\Widebar{\accentset{{\cc@style\underline{\mskip13mu}}}}
\makeatother

\newtheorem{theorem}{Theorem}
\newtheorem{remark}{Remark}

\newtheorem{definition}{Definition}
\newtheorem{proposition}{Proposition}

\newtheorem{example}{Example}

\begin{document}

\captionsetup[figure]{labelfont={ }, name={Fig.}, labelsep=period} 
\pagestyle{empty}

\title{Age of Information Upon Decisions}
\author{\authorblockN{Yunquan Dong\authorrefmark{1},
                                           Zhengchuan Chen\authorrefmark{2},
                                           Shanyun Liu\authorrefmark{3},
                                           and Pingyi Fan\authorrefmark{3} \\
\authorblockA{ \normalsize
\authorrefmark{1}School of Electronic \& Information Engineering, \\
                                 Nanjing University of Information Science \& Technology, Nanjing, China \\
\authorrefmark{2}College of Communication Engineering, Chongqing University, Chongqing, China \\
\authorrefmark{3}
                                    Department of Electronic Engineering, Tsinghua University, Beijing, China\\
yunquandong@nuist.edu.cn 
}
}
}

\maketitle
\thispagestyle{empty}

\vspace{-15mm}
\begin{abstract}
We consider an $M/M/1$ update-and-decide system where Poisson distributed decisions are made based on the received updates.
    We propose to characterize the freshness of the received updates at decision epochs with \textit{Age upon Decisions (AuD)}.
Under the first-come-first-served policy (FCFS), the closed form average AuD is derived.
    We show that the average AuD of the system is determined by the arrival rate and the service rate, and is independent of the decision rate.
Thus, merely increasing the decision rate does not improve the timeliness of decisions.
    Nevertheless, increasing the arrival rate and the service rate simultaneously can decrease the average AuD efficiently.
\end{abstract}


\section{Introduction}
The development of modern information technology has spawned many applications with stringent delay requirements.
     In smart vehicular networks~\cite{Vnet-1-2011, Vnet-2-2011}, for example, vehicles need to share their status (e.g., position, speed, acceleration) timely to ensure safety.
For these scenarios, neither of the traditional measures like delay or throughput is suitable~\cite{Kam-AoI-2015}.
    Note that when delay is small, the received update may not be fresh if the updates come very infrequently;
when throughput is large, the received updates may also be not fresh if they undergo large queueing delay during the transmission process.

    To convey the freshness of received information, therefore, a new metric was proposed in~\cite{Yates-2012-age}, i.e., \textit{age of information (AoI)}.
Specifically, AoI is defined as the elapsed time since the generation of the latest received update~\cite{Yates-2012-age}, i.e., the age of the newest update at the receiver.
    This insightful measure of information freshness has been exhaustively studied in various queueing systems, e.g., $M/M/1, M/D/1$ and $D/M/1$~\cite{Yates-2012-age}, under several serving disciplines, e.g., first-come-first-served (FCFS)~\cite{Yates-2012-age, Dong-2018-two-way},~ last-generate-first-served (LGFS)~\cite{Sun-2016-mlt-sver}, and in multi-source~\cite{Yate-2016-multisource}, multi-class~\cite{Huang-2015-multiclass}, multi-hop~\cite{Sun-2017-multihop} scenarios.

With the following observations that
\begin{itemize}
  \item \textit{delay} quantifies the freshness of updates at the epochs when they are received;
  \item \textit{AoI} quantifies the freshness of updates at every epoch after they are received;
  \item in many update-and-decide systems, the freshness of updates are only important for some decision epochs,
\end{itemize}
we are motivated to consider a new freshness measure termed as \textit{age upon decisions (AuD)}.
    That is,
\begin{itemize}
  \item \textit{AuD} quantifies the freshness of the received updates at those decision epochs when they are used.
\end{itemize}

In particular, AuD can readily be applied to parallel computing based machine learning systems, Internet-of-Things (IoT), cognitive networks, and so on.
\begin{example}
    The AlphaGo system performs Monte Carlo tree search with 1920 CPUs and 280 GPUs in a distributed and parallel manner~\cite{Wiki-2018-alphago}.
In this kind of large-scale parallel computing systems with depth first tree searching, random polling is a simple yet effective dynamic load balancing scheme~\cite{Peter-1994-polling}.
    AuD can then be used to evaluate the utility of the system by characterizing the waiting time from the beginning of  a busy period (update arrival) of a server to the polling epoch (decision epoch) from an idle server.
\end{example}
\begin{example}
In cognitive systems, a secondary user accesses wireless channel by sensing the channel randomly~\cite{Yuan-2012-Cognitive}.
    In this case, AuD is the time elapsed from the beginning of an idle channel period to the sensing epoch (decision epoch) of the secondary user.
Thus, low AuD implies high channel utilization.
\end{example}
\begin{example}
    In large scale wireless sensor networks, IoT networks, and underwater networks, collecting information by random polling can improve system efficiency by avoiding uplink collisions~\cite{Zorzi-2012-AUV}.
        Since random polling epochs (decision epochs) can never be consistent with information generation epochs, AuD is useful to evaluate the timeliness of the information collecting process.
\end{example}

In this paper, therefore, we are interested in the age of information at decision epochs and shall apply it to an $M/M/1$ update-and-decide system.
    We assume that random decisions are made following a Poisson process.
It is surprising to observe that the average AuD of the system is independent of the rate of decisions.
    That is, making more decisions does not help to reduce the average AuD.

This paper is organized as follows.
   Section~\ref{sec:2_model} presents the system model and the definition of AuD.
 We investigate the  average AuD in an $M/M/1$ queueing system in Section~\ref{sec:3_aud} and present the obtained results via numerical simulation in Section~\ref{sec:4_simulation}.
Finally,  our work is concluded in Section~\ref{sec:5_conclusion}.

\section{System Model}\label{sec:2_model}

We consider a FCFS $M/M/1$ update-and-decide system with arrival rate $\lambda$ and service rate $\mu$, as shown in Fig. \ref{fig:net_model}.
    The arrived updates are stored in an infinite long buffer and will be served according to FCFS discipline.
We assume that the server utilization is smaller than unity, i.e., $\rho=\tfrac\lambda\mu<1$, so that the queueing process is stable.
    Based on the received updates, the receiver makes random decisions at rate $\nu$.

As shown in Fig. \ref{fig:aud}, the updates are generated at \textit{arrival epochs} $\{t_k, k=1,2,\cdots\}$ and are received at \textit{departure epochs} $t'_k$.
    The \textit{inter-arrival time} $X_k$ between neighboring updates is $X_{k}=t_k-t_{k-1}$ and the \textit{system time} that packet $k$ stays in the system is $T_{k}=t'_k-t_k$.
Note that  system time is the sum of \textit{waiting time}  $W_{k}$ and  service time $S_{k}$,  i.e., $T_{k} =W_{k}+S_{k}$.
    We denote the period between two consecutive departure epochs as \textit{inter-departure time} $Y_k=t_k'-t_{k-1}'$ and denote the period between two consecutive decision epochs as \textit{inter-decision time} $Z_j=\tau_j-\tau_{j-1}$.
In this paper, we assume that $X_k$, $S_k$, and $Z_k$ are all exponentially distributed random variables with mean $\mathbb{E}[X]=\frac1\lambda$, $\mathbb{E}[S]=\frac1\mu$, and $\mathbb{E}[Z]=\frac1\nu$, respectively.

In this paper, we investigate the freshness of the received updates at  decision epochs via \textit{age upon decision}.

\begin{definition}
     (\textit{Age upon decision-AuD}). At the $j$-th decision epoch, the index of the most recently received update is
          \begin{equation} \nonumber
                N_\text{U}(\tau_j) = \max\{ k|t_k'\leq \tau_j \},
            \end{equation}
      and the generation time of the update is
          \begin{equation} \nonumber
                U(\tau_j) = t_{N_\text{U}(\tau_j)}.
            \end{equation}

      The \textit{Age upon decision} of the update-and-decide system is then defined as the random process
    \begin{equation} \label{df:aud}
        \Delta_\text{D}(\tau_j) = \tau_j - U(\tau_j).
    \end{equation}

\end{definition}

Note that if we replace  decision epochs $\tau_j$ with arbitrary time $t$, AuD $\Delta_\text{D}(\tau_j) $  reduces to AoI $\Delta(t) $.

\begin{figure}[!t]
\centering
\includegraphics[width=2.7in]{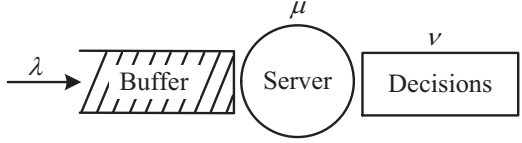}
\caption{The queueing model. } \label{fig:net_model}
\end{figure}

\vspace{3mm}
\begin{figure}[!t]
\centering
\includegraphics[width=3.1in]{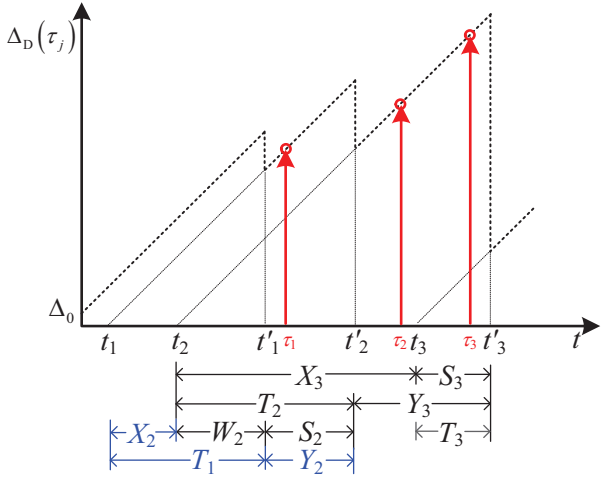}
\caption{Age upon decisions. } \label{fig:aud}
\end{figure}

\begin{example}\label{eg_1}
    Fig.~\ref{fig:aud} shows a sample path of AoI and AUD.
        Since the service of the first update is not completed until $t_1'$, the second update sees a busy server upon its arrival at $t_2$.
  The second update waits for a period of $W_2$ and starts its service immediately at the departure of the first update.
Thus, the inter-departure time $Y_2$ between the first and the second updates is equal to the service time of the second update, i.e., $Y_2=S_2$.
     This is a typical case where  $X_k<T_{k-1}$ is true and we have $Y_k=S_k$.
  On the other hand, if $X_k>T_{k-1}$ is true (e.g., $X_3>T_2$), the next update has not arrived at the departure of update $k$.
    As shown in Fig.~\ref{fig:aud}, the server will be idle for a period of $X_3-T_2$ before the third update gets served from its arrival.
  In this case, the inter-departure time is given by $Y_k=X_k+S_k-T_{k-1}$.

  During each inter-departure time, several decisions can be made based on the received update.
        For example, there are two decision epochs $\tau_2$ and $\tau_3$ (denoted by the red arrows) during $Y_3$ and the corresponding AuD are $\Delta_\text{D}(\tau_2)$ and  $\Delta_\text{D}(\tau_3)$, respectively. $\hfill{} \blacksquare$
\end{example}

For the given arrival process, the serving process, and the decision process, we are interested in the \textit{average AuD} of the system.
    Suppose there are $N_T$ decisions during a period of $T$,  the average AuD is given by
    \begin{equation} \label{df:aud}
        \widebar{\Delta}_\text{D} =\lim_{T\rightarrow\infty}\tfrac{1}{N_T} \sum_{j=1}^{N_T} \Delta_\text{D}(\tau_j),
    \end{equation}
    with $\lim_{j\rightarrow\infty}\tau_j=+\infty$.

\section{Average Age upon Decisions} \label{sec:3_aud}
    In this section, we first investigate the queueing process of the system, and then derive the average AuD closed form.

\subsection{Queueing Process} \label{subsuc:2_A}
At time $t$, we denote the number of updates in the queue as $L(t)$.
    Since server utilization $\rho=\tfrac\lambda\mu$ is smaller than unity, the queue is stable and queue length $L(t)$ has a stationary distribution $\boldsymbol{\pi}=[\pi_0, \pi_1, \pi_2,\cdots]$.
By using the equilibrium equation $\lambda \pi_i=\mu \pi_{i+1}$ and the regularization condition $\sum_{i=0}^\infty \pi_i = 1$, we have
    \begin{equation} \label{rt:pi_L}
        \pi_i =(1-\rho) \rho^i,\quad i=0,1,\cdots.
    \end{equation}

Based on this result, the probability density function (\textit{p.d.f.}) of system time $T_k$ can be given by the following proposition, which is very useful in characterizing inter-departure time $Y_k$.
\begin{proposition}\label{prop:fx_Tk}
    The \textit{p.d.f.} of system time $T_k$ is
        \begin{equation} \label{rt:pdf_Tk}
            f_\text{T}(x) =  \mu(1-\rho) e^{-\mu(1-\rho)x}, \quad x\geq0.
        \end{equation}
\end{proposition}
That is, $T_k$ is exponentially distributed.

\begin{proof}
    See Appendix \ref{apx:f_Tk}.
        Although this result can also be found in other references, e.g.,~\cite{Yates-2012-age}, we present a brief proof here to make it easy to follow.
\end{proof}

Since the departure time $t_k'$ can be expressed as
\begin{equation*}
    t_k' =\sum_{i=1}^k X_i +T_k,
\end{equation*}
where $X_1=t_1$,
     the inter-departure time $Y_k=t_k'-t_{k-1}'$ can be rewritten as
\begin{equation}
    Y_k = X_k +T_k - T_{k-1},\qquad k\geq2.
\end{equation}

In particular, $Y_k$ follows the same distribution as inter-arrival time $X_k$, as shown in the following proposition.
\begin{proposition}\label{prop:mgf_Yk}
    Inter-departure time $Y_k$ is an exponentially distributed random number with rate $\lambda$ and \textit{p.d.f.}
        \begin{equation}\label{rt:mgf_Yk}
                f_\text{Y}(x)=\lambda e^{-\lambda x}.
        \end{equation}
\end{proposition}

\begin{proof}
    See Appendix \ref{apx:mgf_Yk}.
\end{proof}

By considering the number of decision epochs during each inter-departure time, average AuD of the system can be obtained, as shown in the following theorem.
\begin{theorem}\label{thm:thm_1_aud}
    In an $M/M/1$ update-and-decide system with arrival rate $\lambda$, service rate $\mu$, and Poisson decisions at rate $\nu$,  the average AuD of the system is independent of decision rate $\nu$.
        Specifically, the average AuD is given by
    \begin{equation}\label{rt:thm_aud}
        \widebar{\Delta}_\text{D} =\frac1\mu\left(1+ \frac1\rho +\frac{\rho^2}{1-\rho}\right).
    \end{equation}
\end{theorem}

\begin{proof}
    See Appendix \ref{apx:prf_1_aud}.
\end{proof}

\begin{remark}
    From Theorem \ref{thm:thm_1_aud}, we have the following observations.
    \begin{enumerate} []
      \item For our model, the average AuD depends only on arrival rate $\lambda$ and service rate $\mu$, and is independent of decision rate $\nu$.
                        This means that making decisions more frequently does not improve the timeliness of decisions.
                This is because when decision rate is increased, although there will be more decision epochs being closer to the departure time of the newest update, there will also be more decision epochs being farther from the departure time.
                    In the statistical sense, therefore, average AuD does not change with the frequency of decisions.
      \item  If the arrival rate is small, although the waiting time and the system time of updates are small, the average AUD is large since the inter-arrival time is large.
                        On the other hand, although the inter-arrival time is small when the arrival rate is large, the system time would be large due to queueing delays, resulting to large average AuD.
                  To minimize the average AuD, therefore, the arrival rate should neither be too small nor too large.
      \item For a given service rate $\mu$, the optimal arrival rate minimizing $\rho^*$ would be close to $0.5$, i.e., is $\lambda^*\thickapprox\tfrac \mu2$.
      \item For a given arrival rate $\lambda$, the average AuD decreases with service rate $\mu$ and approaches zero as $\mu$ goes to infinity, i.e, updates are transmitted with zero service time.
    \end{enumerate}

Moreover, the conclusion that average AuD is independent of decision rate also applies to general $G/G/1$ queues if decisions are made uniformly in the statistical sense, e.g. periodically or with exponential inter-decision times.
    Thus, we should be more focused on scheduling the update arrival/service process other than the decision process.

\begin{figure}[htp]   

\hspace{-6 mm}
    \begin{tabular}{cc}
    \subfigure[Average AuD versus arrival rate $\lambda$]
    {
    \begin{minipage}[t]{0.5\textwidth}
    \centering
    {\includegraphics[width = 3.5in] {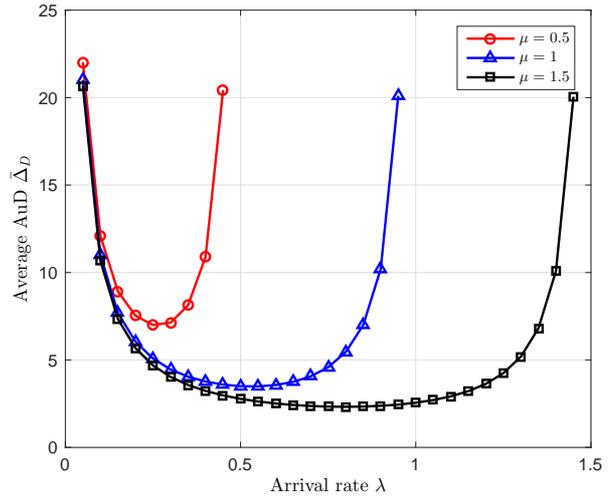} \label{fig:aud_lambda}}
    \end{minipage}
    }\\

    \subfigure[Average AuD versus service rate $\mu$]
    {
    \begin{minipage}[t]{0.5\textwidth}
    \centering
    {\includegraphics[width = 3.5in] {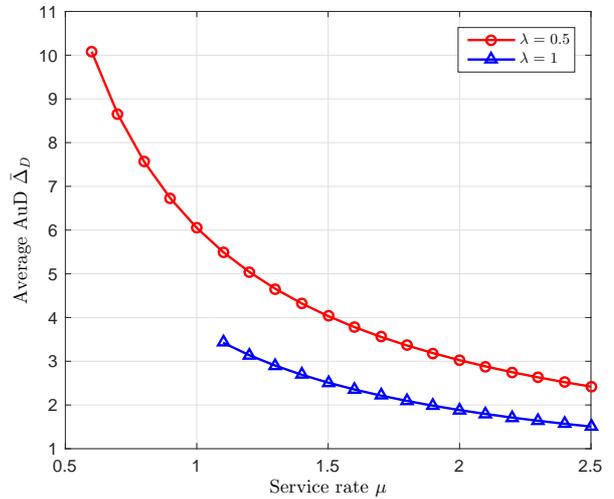} \label{fig:aud_mu}}
    \end{minipage}
    }
%

    \end{tabular}

\caption{Average AuD. } \label{fig:aud_sim}
\end{figure}

\end{remark}

\section{Numerical Results}\label{sec:4_simulation}
In this section, we investigate the average AuD of the $M/M/1$ system via numerical results.
    First, we set the service rate to be a constant (e.g., $\mu=0.5, \mu=1, \mu=1.5$) and investigate how average AuD changes with arrival rate $\lambda$.
As shown in Fig. \ref{fig:aud_lambda}, the average AuD is large when $\lambda$ is either very small or very large.
    To be specific, when $\lambda$ is small, average AuD is large because the waiting time for the arrival of a new update is large.
When $\lambda$ is large, the queueing delay of updates is large due to the limited service capability of the server.
    In this situation, the received updates will be outdated at the decision epochs.
To minimize AuD, therefore, we should try to increase the service rate and set the arrival rate to be a half of service rate, i.e., $\lambda=\mu/2$.

Fig. \ref{fig:aud_mu} presents how average AuD changes with service rate for a given arrival rate.
    As is shown, average AuD is monotonically decreasing with service rate $\mu$.
Moreover,  given a large $\mu$, AuD is smaller if arrive rate $\lambda$ is larger.
    This means that the AuD performance is better when the arrival rate and the service rate grow larger at the same time.

\begin{figure}[!t]
\centering
\includegraphics[width=3.5in]{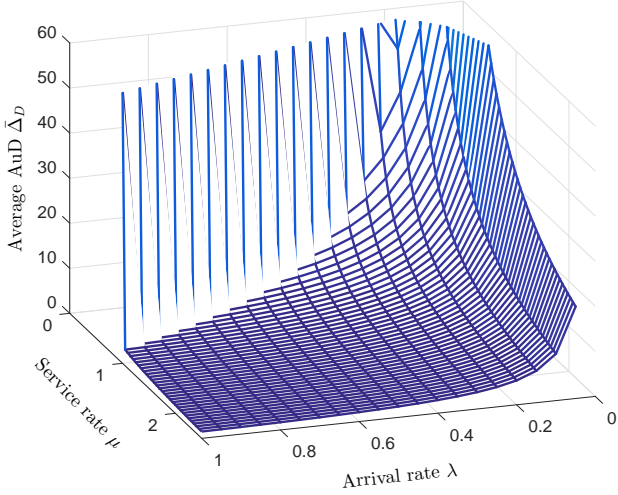}
\caption{Average AuD versus $\lambda$ and $\mu$. } \label{fig:aud_3d}
\end{figure}

\begin{figure}[!t]
\centering
\includegraphics[width=3.5in]{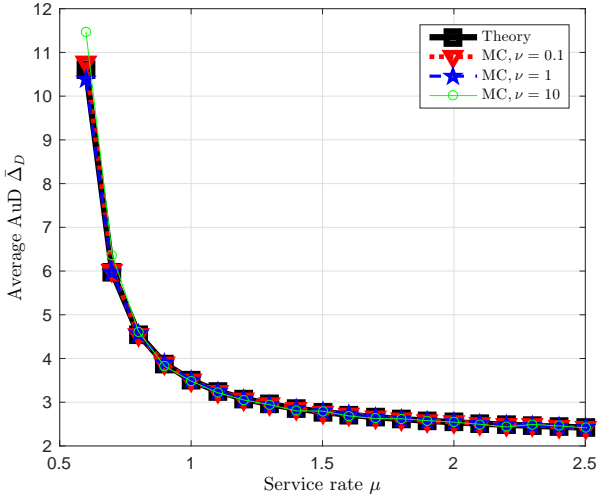}
\caption{Average AuD versus $\mu$ via Monte Carlo simulation. } \label{fig:aud_mu_mc}
\end{figure}

We further plot the variations of average AuD with arrival rate and service rate in Fig. \ref{fig:aud_3d}.
    It is clear that average AuD is small when $\lambda$ and $\mu$ are large.
When they are both small, it is seen that $\mu$ brings more effect on AuD performance.
    In particular, average AuD goes to infinity much faster if $\mu$ is reduced to zero than $\lambda$ is reduced to zero.

Fig. \ref{fig:aud_mu_mc} presents average AuD via Monte Carlo simulations.
    We set arrival rate to $\lambda=0.5$ and consider three decision rates, i.e.,  $\nu=0.1, \nu=1$, and $\nu=10$.
During a period in which $N_\text{T}=10^6$ updates are generated and served, approximately $K=2\times10^5, K=2\times10^6$ and $K=2\times10^7$ decisions are made.
    As shown Fig. \ref{fig:aud_mu_mc}, the Monte Carlo results do not change with decision rate $\nu$ and coincide with the corresponding theory results (see \eqref{rt:thm_aud}).

\section{Conclusion}\label{sec:5_conclusion}
In this paper, we have proposed a new measure termed \textit{age upon decisions} to evaluate the freshness of updates at decision epochs.
    For an $M/M/1$ update-and-decide system with Poisson decision process, we proved that the average AuD of the system is independent the rate of the decision process.
Thus, making decisions more frequently does not improve the timeliness of these decisions.
    Moreover, the proposed AuD measure has many practical applications.
In a cognitive communication system, for example, the secondary user can choose its sensing rate based on its demand on channel uses and its sensing cost, without extra consideration on the timeliness of the accessing time.
    If the decision rate is set to be very small, however, the receiver may miss a lot of received updates.
Therefore, characterizing system performance jointly with AuD and update missing probability would be an interesting extension of this work.

\appendix

\renewcommand{\theequation}{\thesection.\arabic{equation}}
\newcounter{mytempthcnt}
\setcounter{mytempthcnt}{\value{theorem}}
\setcounter{theorem}{2}
\addcontentsline{toc}{section}{Appendices}\markboth{APPENDICES}{}

\subsection{Proof of Proposition \ref{prop:fx_Tk}} \label{apx:f_Tk}
\begin{proof}
Suppose that at the arrival of update $k$, the number of updates in the queue is $L(t)=i$.
    It is clear that update $k$ will not get served until all waiting updates are completed, i.e., $W_k=\sum_{j=2}^{i} S_{(j)} +S_{(1)}^\text{re}$, where $S_{(1)}^\text{re}$ is the remaining service time of current update.
Since $S_{(1)}$ follows the memoryless  exponential distribution, we know that $S_{(1)}^\text{re}$ has the same distributions as $S_{(1)}$.
    Thus, system $T_k$ can be rewritten as $T_k=W_k+S_k=\sum_{j=1}^{i+1} S_{(j)}$, which follows Erlang distribution $\text{Erlang}(i+1, \mu)$.
The probability that system time $T_k$ is larger than $x$ can expressed as
    \begin{eqnarray*}
        \Pr\{T_k>x\} \hspace{-2.75mm} & = &\hspace{-2.75mm} \sum_{i=0}^\infty \Pr\left\{\sum_{j=1}^{i+1} S_{(j)}>x \right\} \Pr\{L(t)=i\} \\
                    \hspace{-2.75mm} & = &\hspace{-2.75mm} \sum_{i=0}^\infty \sum_{j=0}^{i} \tfrac{1}{j!}e^{-\mu x}{(\mu x)}^j \cdot (1-\rho) \rho^i \\
                    \hspace{-2.75mm} & = &\hspace{-2.75mm} \sum_{j=0}^\infty \sum_{i=j}^{\infty} \tfrac{1}{j!}e^{-\mu x}{(\mu x)}^j \cdot (1-\rho) \rho^i \\
                    \hspace{-2.75mm} & = &\hspace{-2.75mm} e^{-\mu(1-\rho)x}, \quad x\geq0.
    \end{eqnarray*}

    The \textit{p.d.f.} of $T_k$ can then be readily obtained by taking the derivative of $1-\Pr\{T_k>x\} $, which completes the proof of \textit{Proposition} \ref{prop:fx_Tk}.
\end{proof}

\subsection{Proof of Proposition \ref{prop:mgf_Yk}} \label{apx:mgf_Yk}
\begin{proof}
As  discussed in \textit{Example} \ref{eg_1},  inter-departure time $Y_k$ is given by
\begin{equation}\label{dr:apx_mgf_Yk}
    Y_k =  \left\{
                            \begin{aligned}
                                    &S_k,    &\text{if}~X_k<T_{k-1}\\
                                    &X_k+S_k-T_{k-1},   &\text{if}~X_k>T_{k-1}.
                            \end{aligned}
                \right.,
\end{equation}
and thus we have,
\begin{eqnarray*}
G_\text{Y}(s) \hspace{-2.75mm}&=&\hspace{-2.75mm} \Pr\{X_k>T_{k-1}\} \mathbb{E}[e^{sY_k}]  + \Pr\{X_k<T_{k-1}\} \mathbb{E}[e^{sY_k}]  \\
                                   \hspace{-2.75mm}&=&\hspace{-2.75mm} \Pr\{X_k>T_{k-1}\} \mathbb{E}[e^{s(X_k-T_{k-1})}] \mathbb{E}[e^{sS_k}]  \\
                                   \hspace{-2.75mm}&&\hspace{-2.75mm}     +  \Pr\{X_k<T_{k-1}\} \mathbb{E}[e^{sS_k}],
\end{eqnarray*}
where the second equation follows \eqref{dr:apx_mgf_Yk} and the fact that $S_k$ is independent with  $X_k$ and $T_{k-1}$.

First, the probability that inter-arrival time $X_k$ is smaller than  previous system time $T_{k-1}$ can be obtained as follows.
\begin{equation*}
\Pr\{X_k<T_{k-1}\}  = \int_0^\infty f_\text{X}(x) \text{d}x \int_x^\infty f_\text{T}(t)\text{d} t  = \rho,
\end{equation*}
where $f_\text{X}(x)=\lambda e^{-\lambda x}$ is the \textit{p.d.f.} of $X_k$ and $ f_\text{T}(t)$ is given by \textit{Proposition} \ref{prop:fx_Tk}.
    We also have $\Pr\{X_k>T_{k-1}\}  =1-\rho$.

Second, the MGF of $Y_k$ conditioned on $X_k>T_{k-1}$ is given by
\begin{eqnarray*}
\hspace{-2.75mm}&&\hspace{-2.75mm} \mathbb{E}[e^{sY_k}|X_k>T_{k-1}] =\mathbb{E}[e^{s(X_k-T_{k-1})}|X_k>T_{k-1}] \mathbb{E}[e^{sS_k}]   \\
                                   \hspace{-2.75mm}&&\hspace{-2.75mm} =\int_0^\infty f_\text{T}(t)\text{d}t \int_t^\infty f_{X|X>t}(x)e^{s(x-t)} \text{d}x
                                                                                                                    \int_0^\infty f_\text{S}(x) e^{sx}\text{d}x \\
                                   \hspace{-2.75mm}&&\hspace{-2.75mm} = \frac{\lambda\mu}{(\lambda-s)(\mu-s)},
\end{eqnarray*}
where $f_\text{T}(t)$ is given by \text{Proposition} \ref{prop:fx_Tk}, $f_\text{S}(x) =\mu e^{-\mu x}$ is the \textit{p.d.f.} of service time $S_k$, and
\begin{equation*}
         f_{X|X>t}(x) = \frac{ f_{X}(x)}{\Pr\{X>t\}}=\lambda e^{-\lambda(x-t)}, \quad x>t.
\end{equation*}

Also note that
\begin{equation*}
         \mathbb{E}[e^{sY_k}|X_k<T_{k-1}] = \mathbb{E}[e^{sS_k}]  = \frac{\mu}{\mu-s}
\end{equation*}
Combining the obtained results, the proof of the proposition would be completed readily.

\end{proof}

\subsection{Proof of Theorem \ref{thm:thm_1_aud}}  \label{apx:prf_1_aud}
\begin{proof}
    Given an inter-departure time $Y_k=y$, suppose $N_k$ decisions are made at epochs $\{\tau_j, j=1,2,\cdots,N_k\}$.
        It is clear that $N_k$ is a Poisson distributed random number with parameter $\nu y$.
    That is, the probability that $n$ decisions are made during $Y_k$ is
\begin{equation*}
    \Pr\{N_k=n|Y_k=y\}=\tfrac{(\nu y)^n}{n!}e^{-\nu y}.
\end{equation*}

We denote $\tau'_j=\tau_j-t'_{k-1}$.
    Since decision epochs $\tau_j$ are independently and uniformly distributed in $Y_k$,  $\tau'_j$ would be independently and uniformly distributed over $[0,y]$.
Thus, the expected sum $\Delta'_{\text{D}k}=\sum_{j=1}^n \tau'_j$   can be  expressed as
\begin{equation*}
    \mathbb{E}\big[\Delta'_{\text{D}k}|Y_k=y,N_k=n\big] =\sum_{j=1}^n \mathbb{E}[\tau'_j]=\tfrac{ny}{2}.
\end{equation*}

Note that the AuD at decision epoch $\tau_j$ is $\Delta_{\text{D}k}(\tau_j)=T_{k-1}+\tau'_j$, where $T_{k-1}$ is the system time of the latest received update.
    We then have
\begin{eqnarray*}
    \mathbb{E}[\Delta_{\text{D}k}|Y_k=y]
    \hspace{-2.75mm}&=&\hspace{-2.75mm}  \sum_{n=0}^\infty \Pr\{N_k=n|Y_k=y\} (\tfrac{ny}{2}  + n\mathbb{E}[T_{k-1}]) \\
    \hspace{-2.75mm}&=&\hspace{-2.75mm}  \nu  (\tfrac{y^2}{2} + y\mathbb{E}[T_{k-1}]).
\end{eqnarray*}

Taking the expectation over $Y_k$, we have
\begin{eqnarray*}
    \mathbb{E}[\Delta_{\text{D}k}] =\tfrac{\nu}{2}\mathbb{E}[Y_k^2] + \nu\mathbb{E}[T_{k-1}Y_k],
\end{eqnarray*}
    where $\mathbb{E}[T_{k-1}Y_k]$ is given by
\begin{eqnarray*}
    \mathbb{E}[T_{k-1}Y_k]\hspace{-2.75mm}&=&\hspace{-2.75mm} \mathbb{E}[T_{k-1}Y_k|T_{k-1}>X_k]\Pr\{T_{k-1}>X_k\} \\
                                    \hspace{-2.75mm}&&\hspace{-2.75mm} + \mathbb{E}[T_{k-1}Y_k|T_{k-1}\leq X_k]\Pr\{T_{k-1}\leq X_k\} \\
                                    \hspace{-2.75mm}&=&\hspace{-2.75mm}  (1-\rho)\mathbb{E}[T_{k-1}X_k-T_{k-1}^2|T_{k-1}\leq X_k] \\
                                    \hspace{-2.75mm}&=&\hspace{-2.75mm} +\mathbb{E}[T_{k-1}]\mathbb{E}[S_{k}].
\end{eqnarray*}

Based on \textit{Proposition} \ref{prop:fx_Tk}, we further have
\begin{eqnarray*}
    \hspace{-6mm}&&\mathbb{E}[T_{k-1}X_k-T_{k-1}^2|T_{k-1}\leq X_k]\\
    \hspace{-6mm}&& \frac{1}{1-\rho} \int_{0}^\infty f_\text{X} (x) \text{d}x \int_0^x (xt-t^2) f_\text{T}(x)\text{d} t\\
    \hspace{-6mm}&& \frac{1}{1-\rho} \frac{1-\rho}{\mu^2\rho},
\end{eqnarray*}
and hence
\begin{equation}
    \mathbb{E}[T_{k-1}Y_k] = \frac{1}{\mu^2(1-\rho)}+\frac{1-\rho}{\mu^2\rho}.
\end{equation}

Assume that there are $K$  departure epochs and $N_\text{T}$ decision epochs during a period $T$, we have $N_\text{T}=\sum_{k=1}^K N_k$.
    As $T$ goes to infinity, we have
\begin{eqnarray}
        \nonumber \widebar{\Delta}_\text{D} \hspace{-2.75mm}&=&\hspace{-2.75mm}
                                \lim_{T\rightarrow\infty} \tfrac{1}{N_\text{T}}\sum_{k=1}^K \Delta_{\text{D}k}
                                = \lim_{T\rightarrow\infty} \tfrac {K}{N_\text{T}} \tfrac1K\sum_{k=1}^K \Delta_{\text{D}k}
                                = \tfrac{\mathbb{E}[\Delta_{\text{D}k}] }{\nu\mathbb{E}[Y_k]} \\
       \label{apx:rt_aud}  \hspace{-2.75mm}&=&\hspace{-2.75mm} \frac{\mathbb{E}[Y_k^2]+2\mathbb{E}[T_{k-1}Y_k]}{2\mathbb{E}[Y_k]}.
\end{eqnarray}

From \textit{Proposition} \ref{prop:mgf_Yk} and \textit{Proposition} \ref{prop:fx_Tk}, we know that $\mathbb{E}[Y_k]=\frac1\lambda$ and $\mathbb{E}[Y_k^2]=\frac{2}{\lambda^2}$.
    Inserting $\mathbb{E}[Y_k]$,  $\mathbb{E}[Y_k]$,   $\mathbb{E}[Y_k]$, and $\mathbb{E}[T_{k-1}Y_k]$ into \eqref{apx:rt_aud},  the proof of \textit{Theorem} \ref{thm:thm_1_aud} would be completed.
\end{proof}


\small{
\bibliographystyle{IEEEtran}

\begin{thebibliography}{11}


\bibitem{Vnet-1-2011}
S. Kaul, M. Gruteser, V. Rai, and J. Kenney, ``Minimizing age of information in vehicular networks," in \textit{Proc. IEEE SECON}, Salt Lake, Utah, USA, Jun. 2011, pp. 350--358.

\bibitem{Vnet-2-2011}
S. Kaul, R. Yates, and M. Gruteser, ``On piggybacking in vehicular networks,” in \textit{Proc. IEEE Global Tel. Conf. (GLOBECOM)}, Houston, Texas, USA, Dec. 2011, pp. 1--5.

\bibitem{Kam-AoI-2015}
C. Kam, S. Kompella and A. Ephremides, ``Experimental evaluation of the age of information via emulation," in \textit{Proc. IEEE Military Commun. Conf. (MILCOM)}, Tampa, FL., USA, Oct. 2015, pp. 1070--1075.

%
%

\bibitem{Yates-2012-age}
S. K. Kaul, R. D. Yates, and M. Gruteser, ``Real-time status: How often should one update?" in \textit{Proc. IEEE Int. Conf. Comput. Commun. (INFOCOM)}, Orlando, FL, USA, Mar. 2012, pp. 2731--2735.

\bibitem{Dong-2018-two-way}
Y. Dong, Z. Chen, and P. Fan, ``Uplink age of information of unilaterally powered two-way data exchanging systems," in \textit{Proc. IEEE Int. Conf. Comput. Commun. Wkshp. (INFOCOM Wkshp)}, Honolulu, HI, US, Apl. 2018, pp. 559--564.


\bibitem{Sun-2016-mlt-sver}
A. M. Bedewy, Y. Sun, and N. B. Shroff, ``Optimizing data freshness, throughput, and delay in multi-server information-update systems," \textit{Proc. IEEE Int. Symp. Inf. Theory (ISIT)}, Barcelona, Spain, Jul. 2016, pp.  2569--2573.

%
%
%
%
%
%

\bibitem{Yate-2016-multisource}
R. D. Yates and  S. K. Kaul, ``The age of information:  Real-time status updating by multiple sources," [Online]. Available: arXiv:1608.08622v1.

\bibitem{Huang-2015-multiclass}
L. Huang and E. Modiano, ``Optimizing age-of-information in a multi-class queueing system,"  in \textit{Proc. IEEE Int. Symp. Inf. Theory} (\textit{ISIT}),  Hong Kong, China, Jun. 2015, pp. 1681--1685.

\bibitem{Sun-2017-multihop}
A. M. Bedewy, Y. Sun and N. B. Shroff, ``Age-optimal information updates in multihop networks," in \textit{Proc. IEEE Int. Symp. Inf. Theory (ISIT)}, Aachen, Germany, Jun. 2017, pp.  576--580.

\bibitem{Wiki-2018-alphago}
Wikipedia contributors. AlphaGo. Wikipedia, The Free Encyclopedia. May 21, 2018, 10:55 UTC. Available at: \url{https://en.wikipedia.org/w/index.php?title=AlphaGo&oldid=842270143}. Accessed  June 4, 2018.

\bibitem{Peter-1994-polling}
P. Sanders, ``A detailed analysis of random polling dynamic load balancing," in \textit{Proc. Int. Symp. Parallel Archit., Algorithms. Netw. (ISPAN)}, Kanazawa, Japan, Dec. 1994, pp.  382--389.

\bibitem{Yuan-2012-Cognitive}
Q.  Liang, M. Liu, D. Yuan, ``Channel estimation for opportunistic spectrum access: Uniform and random sensing," \textit{IEEE Trans. Mobile Comput.}, vol. 11, no. 8, pp. 1304--1316, Aug. 2012.

\bibitem{Zorzi-2012-AUV}
F. Favaro, P. Casari, Federico Guerra, M. Zorzi, ``Data upload from a static underwater network to an AUV: Polling or random access?" in \textit{Oceans}, Yeosu, South Korea, May 2012, pp. 1--6.






%

\end{thebibliography}

}
\end{document}